\documentclass[conference]{IEEEtran}
\IEEEoverridecommandlockouts
\usepackage{graphicx,amsmath,amssymb,amsfonts,cite}
\usepackage{subfigure}
\usepackage{algorithm}
\usepackage{algorithmic}
\usepackage{float}
\usepackage{amsthm}
\usepackage[english]{babel}
\usepackage{ulem}
\usepackage{bm}
\usepackage{multirow}
\usepackage{array}
\usepackage{color}
\usepackage{ulem}
\usepackage{caption}
\usepackage{flushend}
\ifCLASSINFOpdf
\else
\fi
\begin{document}
%
\title{Optimal Preamble Length for Spectral Efficiency in Grant-Free RA with Massive MIMO
}

\author{\IEEEauthorblockN{Jie Ding, Daiming Qu, and Hao Jiang }
\IEEEauthorblockA{
School of Electronic Information and Communications\\
Huazhong University of Science and Technology, Wuhan, P. R. China\\
Email: yxdj2010@gmail.com}
\thanks{This work was supported in part by the National Natural Science Foundation of China under Grant 61701186 and Grant 61571200 and in part by the China Postdoctoral Science Foundation funded project under Grant 2017M612458.}}

\maketitle

\begin{abstract}
Grant-free random access (RA) with massive MIMO is a promising RA technique for massive access with low signaling overhead. In the grant-free
RA with massive MIMO, preamble length has a critical impact on the performance of the system. In this paper, the optimal preamble length is investigated to maximize spectral efficiency (SE) of the grant-free RA with massive MIMO, where effects of the preamble length on the preamble collision and preamble overhead as well as channel estimation accuracy are taken into account. Simulation results agree well with our analyses and confirm the existence of optimal preamble length for SE maximization in the grant-free RA with massive MIMO.
Moreover, properties of the optimal preamble length with respect to system parameters are revealed. Compared to the granted access, it is shown that longer preamble length is required for SE maximization in the grant-free RA. 

\end{abstract}


%
\IEEEpeerreviewmaketitle

\section{Introduction}
 Future wireless communication systems not only aim to enhance the traditional mobile broadband use case, but also to meet the requirements of new emerging use cases, such as Internet of Things (IoT) \cite{1}\cite{2}. As an enabler of the IoT, machine-to-machine (M$2$M) communications have been a topic of interest for researchers in recent years. In M$2$M, the amount of random access (RA) user equipments (UEs) is massive and their payloads are usually short and sporadic in nature. Therefore, fulfilling the demand of massive access with low signaling overhead and access delay is a key technological issue for M$2$M communications \cite{3}.

The legacy request-grant RA in long term evolution (LTE) was only designed to provide reliable access to a small number of UEs with long packets to transmit \cite{4}. Considering small-sized packets in M$2$M communications, however, the request-grant RA becomes not cost-effective because it brings in relatively long waiting time before data transmissions for RA UEs. To support M$2$M communications, several modifications and improvements have been proposed \cite{8, 9, 34}.
Nevertheless,
this type of LTE based RA still suffers from wireless resource scarcity, which is a fundamental bottleneck for enabling massive access.

To effectively manage M$2$M communications, grant-free RA (also known as one-stage RA) with massive multiple-input multiple-output (MIMO) is now being considered as a compelling alternative \cite{20}. In the grant-free RA, request-grant procedure in the legacy RA is omitted and RA UEs contend (i.e., perform RA) with their uplink payloads directly by transmitting preamble along with data. As a result, signaling overhead and access delay are minimized, and the radio resources reserved for the request-grant procedure could be unleashed for accommodating more RA UEs. On the other hand,
massive MIMO has been a main focus of recent research for its capability of mitigating wireless resource scarcity \cite{14}.
With all the benefits manifested in \cite{16,17,18}, features of massive MIMO could be exploited to effectively accommodating multiple access in grant-free RA.
Therefore, the grant-free RA with massive MIMO exhibits potential advantages towards addressing RA issues for future wireless communications.

The work in \cite{20} confirmed the effectiveness of grant-free RA with massive MIMO in accommodating massive access with low signaling overhead. It also indicated that longer preamble length leads to larger number of orthogonal preambles and lower preamble collision rate, therefore resulting in higher success probability for grant-free RA UEs. Nevertheless, longer preamble length does not necessarily bring on better spectral efficiency (SE) for the grant-free RA with massive MIMO.
In fact, the preamble length has three conflicting effects on the SE of grant-free RA. For instance, with fixed grant-free RA packet length, increasing the preamble length will 1) lower the preamble collision rate; 2) improve the accuracy of channel estimation; but 3) reduce the amount of data symbols. Thus, the preamble length should be investigated for a good tradeoff among preamble collision and
preamble overhead as well as channel estimation accuracy.

In this paper, optimal preamble length is explored to maximize SE in the grant-free RA with massive MIMO. Specifically,
we first derive the average SE in terms of the preamble length in the grant-free RA with massive MIMO, where effects of the preamble length on the preamble collision and preamble overhead as well as channel estimation accuracy are taken into account. Based on the obtained SE, we then analyze and derive the optimal preamble length for its maximization. Simulation results verify the accuracy of our analyses and confirm the existence of optimal preamble length for SE maximization in the grant-free RA with massive MIMO.
Moreover, properties of the optimal preamble length with respect to system parameters are revealed. Compared to the granted access, it is shown that longer preamble length is required for SE maximization in the grant-free RA. 

The remainder of this paper is organized as follows. In Section II,
the grant-free RA with massive MIMO is briefly described. In Section III, the SE in terms of the preamble length is given and derivation of the optimal preamble length for SE maximization is detailed.
Simulation results are
presented in Section IV and the work is concluded in Section V.
\normalem

\emph{Notations: }Boldface lower and upper case symbols represent
vectors and matrices, respectively. $\mathbf{I}_n$ is the $n \times n$ identity matrix. The trace, conjugate, transpose, and complex conjugate transpose
operators are denoted by $\mathrm{tr}(\cdot)$, $(\cdot)^{*}$, $(\cdot)^{\mathrm{T}}$ and $(\cdot)^{\mathrm{H}}$. $\mathbb{E}[\cdot]$ denotes the expectation operator.

\section{Grant-Free Random Access with Massive MIMO}
In this section, the system model with massive MIMO is introduced and the procedure of grant-free RA with massive MIMO is illustrated.

In Fig. \ref{fig1}, a single-cell massive MIMO system is considered, where BS is configured with $M$ active antenna elements and independent single-antenna RA UEs are uniformly distributed in the cell. In each RA slot, we assume that $N$ RA UEs are active to perform the grant-free RA simultaneously over the same channel. 
\captionsetup[figure]{name={Fig.},labelsep=period}
\begin{figure}[!h]
\centering
\includegraphics[width=3.0in]{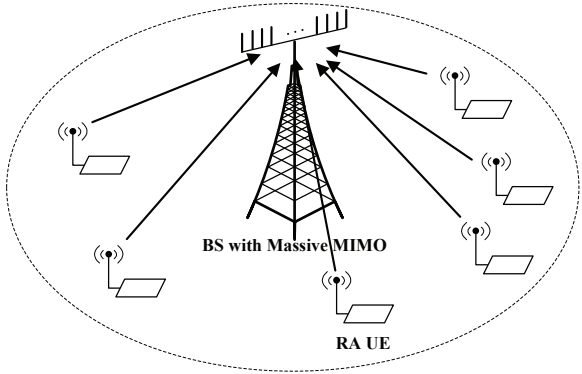}
\caption{Illustration of massive MIMO system model.} \label{fig1}
\end{figure}

Since $M$ is sufficiently large in massive MIMO, favorable propagation (FP) can be approximately achieved. The feature of FP enables spatial multiplexing of multiple RA UEs over the same channel \cite{14}. Specifically, simple linear processing, such as conjugate beamforming (CB) and zero-forcing beamforming (ZFB), could be applied at the BS, to discriminate the signal transmitted by each RA UE from the signals of other RA UEs.
In this paper, CB is employed due to its low detection complexity.
\begin{figure}[!h]
\centering
\includegraphics[width=2.7in]{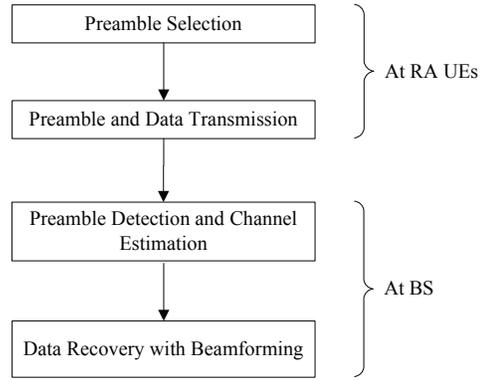}
\caption{Procedure of grant-free RA with massive MIMO.} \label{fig2}
\end{figure}

In Fig. \ref{fig2}, the procedure of grant-free RA with massive MIMO is briefly described. Specifically, the $N$ RA UEs contend with their uplink payloads over the same channel by directly transmitting RA preamble along with data. The RA preamble of each RA UE is randomly chosen out of an RA preamble pool, which is used by the BS for preamble detection and channel estimation. We assume that there are $P$ orthogonal RA preambles available in the pool. If the chosen preamble by an RA UE is different from the ones by other RA UEs, the RA UE could be detected and its channel response could be estimated at the BS. Otherwise, if multiple RA UEs choose the same preamble, preamble collision occurs and we assume that all these RA UEs would not be detected and their channel estimations would be failed.
At the BS, after preamble detection and channel estimation, CB is then used for data recovery.

\begin{figure}[!h]
\centering
\includegraphics[width=2.3in]{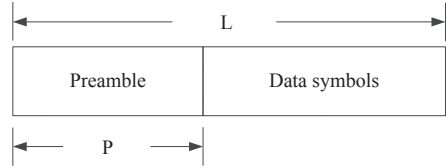}
\caption{Grant-free RA packet structure.} \label{fig3}
\end{figure}

In Fig. \ref{fig3}, grant-free RA packet structure is illustrated. Herein, we assume that the packet length for each RA UE is $L$, where the preamble length is $P$ and the data length is $L-P$. With a fixed $L$, it is plain to see that the preamble length $P$ has a direct effect on the preamble collision and preamble overhead as well as channel estimation accuracy. Since these three performance indicators have different degrees of impact on the SE in the grant-free RA with massive MIMO, it is necessary to investigate the optimal preamble length for SE maximization, which will be detailed in the next section.

\vspace{0.5mm}
\section{SE Maximization }
Taking into account effects of the preamble length on the preamble collision and preamble overhead as well as channel estimation accuracy, expression of average SE for an arbitrary RA UE in terms of the preamble length is derived and the optimal preamble length is obtained accordingly in this section.

In this paper, a block independent Rayleigh fading propagation model is considered, where the propagation channels are assumed constant within length $L$.
The channel response vector between an arbitrary RA UE and the BS is modelled by $\mathbf{g}=\sqrt{\ell}\mathbf{h} \in \mathbb{C}^{M}$, where $\ell$ denotes the large scale fading coefficient between RA UE and BS, and $\mathbf{h} \sim \mathcal{CN}(0,\mathbf{I}_{M})$ stands for the small scale fading vector between RA UE and BS. Moreover, perfect power control is assumed so that all RA UEs have the same expected receive power at the BS.

\subsection{Preamble Collision}

As stated in Section II, when the preamble collision occurs, i.e., multiple RA UEs select the same preamble, the BS is unable to acquire the correct channel responses of these RA UEs and their data transmissions would be failed. Under the condition, their SE is considered as zero as a result. In other words, only RA UEs without experiencing preamble collision have the chance to get their data recovered by the BS. For the $i$th RA UE, its probability that no preamble collision occurs is written by
\begin{align}\label{eq1}
{P}_{\mathrm{no\_collision}}=(1-\frac{1}{P})^{N-1}.
\end{align}

\subsection{Channel Estimation}

At the BS, the received preamble signal $\mathbf{Y} \in \mathbb{C}^{M \times P}$ is
\begin{align}\label{eq2}
\mathbf{Y}=\sum_{i=1}^{N}\sqrt{P_\mathrm{R}}\mathbf{h}_i\mathbf{p}^{\mathrm{T}}_i+\mathbf{N},
\end{align}
where $P_{\mathrm{R}}$ is the expected receive power from each RA UE at each BS antenna. $\mathbf{{h}}_{i}$ is the small scale fading vector between the $i$th RA UE and the BS. $\mathbf{p}_i \in \mathbb{C}^{P}$ is preamble vector transmitted by the $i$th RA UE and $\mathbf{p}^{\mathrm{H}}_i\mathbf{p}_i=P$. $\mathbf{{N}}$ is the noise matrix with i.i.d. elements distributed
as $\mathcal{CN}(0,\sigma^2)$. We denote the uplink signal to noise ratio (SNR) at the BS corresponding to each RA UE by $\rho_{\mathrm{R}}\triangleq P_{\mathrm{R}}/\sigma^2$.

For the $i$th RA UE without experiencing preamble collision, its preamble could be detected successfully thanks to the mutual orthogonality among preambles. With least-squares (LS)-based channel estimation, its estimated channel is given by
\begin{align}\label{eq3}
\mathbf{\hat{h}}_i=\frac{1}{P\sqrt{P_{\mathrm{R}}}}\mathbf{Y}\mathbf{p}_{i}^{*}=\mathbf{{h}}_i+\frac{1}{\sqrt{\rho_{\mathrm{R}}P}}\mathbf{{n}},
\end{align}
where $\mathbf{{n}}=\frac{\mathbf{N}\mathbf{p}_{i}^{*}}{\sigma\sqrt{P}}$ a noise vector distributed as $\mathcal{CN}(0,\mathbf{I}_M)$.

\subsection{Average SE in terms of Preamble Length}

At the BS, the received data signal $\mathbf{r} \in \mathbb{C}^{M}$ is
\begin{align}\label{eq4}
\mathbf{r}=\sum_{i=1}^{N}\sqrt{P_\mathrm{R}}\mathbf{h}_ix_i+\bar{\mathbf{n}},
\end{align}
where $\bar{\mathbf{n}}$ a noise vector distributed as $\mathcal{CN}(0,\sigma^2\mathbf{I}_M)$. $x_{i} $ is data symbol transmitted by the $i$th RA UE and $\mathbb{E}[|x_{i}|^2]=1$.

For the $i$th RA UE without experiencing preamble collision, its detected data symbol with CB is given by
\begin{align}\label{eq5}
\hat{x}_{i}=&\mathbf{\hat{h}}_i^{\mathrm{H}}\mathbf{r} \nonumber\\
=&\sqrt{P_\mathrm{R}}\mathbf{\hat{h}}_i^{\mathrm{H}}\mathbf{h}_ix_i+\sum_{k=1, k\neq i}^{N}\sqrt{P_\mathrm{R}}\mathbf{\hat{h}}_i^{\mathrm{H}}\mathbf{h}_kx_k+\mathbf{\hat{h}}_i^{\mathrm{H}}\bar{\mathbf{n}} \nonumber\\
=&\sqrt{P_\mathrm{R}}\mathbb{E}[\mathbf{\hat{h}}_i^{\mathrm{H}}\mathbf{h}_i]x_i+\sqrt{P_\mathrm{R}}\big(\mathbf{\hat{h}}_i^{\mathrm{H}}\mathbf{h}_i-\mathbb{E}[\mathbf{\hat{h}}_i^{\mathrm{H}}\mathbf{h}_i]\big)x_i\nonumber\\
&+\sum_{k=1, k\neq i}^{N}\sqrt{P_\mathrm{R}}\mathbf{\hat{h}}_i^{\mathrm{H}}\mathbf{h}_kx_k+\mathbf{\hat{h}}_i^{\mathrm{H}}\bar{\mathbf{n}}.
\end{align}

With (\ref{eq5}), the instantaneous signal to interference and noise ratio (SINR) of the $i$th RA UE is given by

\begin{align}\label{eq6}
&\mathrm{SINR}_{i} \nonumber\\
=& \frac{P_\mathrm{R}\mathbb{E}[\mathbf{\hat{h}}_i^{\mathrm{H}}\mathbf{h}_i]^2}{P_\mathrm{R}\big|\mathbf{\hat{h}}_i^{\mathrm{H}}\mathbf{h}_i-\mathbb{E}[\mathbf{\hat{h}}_i^{\mathrm{H}}\mathbf{h}_i]\big|^2
+\sum_{k=1, k\neq i}^{N}P_\mathrm{R}|\mathbf{\hat{h}}_i^{\mathrm{H}}\mathbf{h}_k|^2+|\mathbf{\hat{h}}_i^{\mathrm{H}}\bar{\mathbf{n}}|^2}.
\end{align}
Since $M$ is sufficiently large in massive MIMO, the asymptotic
deterministic equivalence of SINR is obtained as (\ref{eq7}) on the top of the next page \cite{28}.
After some straightforward calculations, it is simplified as
\newcounter{mytempeqncnt}
\begin{figure*}[!t]
\normalsize
\setcounter{mytempeqncnt}{\value{equation}}
\setcounter{equation}{6}
\begin{align}
\label{eq7}
\overline{\mathrm{SINR}}_{i}
= \frac{P_\mathrm{R}\mathbb{E}[\mathbf{\hat{h}}_i^{\mathrm{H}}\mathbf{h}_i]^2}{P_\mathrm{R}\mathbb{E}\Big[\big|\mathbf{\hat{h}}_i^{\mathrm{H}}\mathbf{h}_i-\mathbb{E}[\mathbf{\hat{h}}_i^{\mathrm{H}}\mathbf{h}_i]\big|^2\Big]
+\sum_{k=1, k\neq i}^{N}P_\mathrm{R}\mathbb{E}\Big[|\mathbf{\hat{h}}_i^{\mathrm{H}}\mathbf{h}_k|^2\Big]+\mathbb{E}\Big[|\mathbf{\hat{h}}_i^{\mathrm{H}}\bar{\mathbf{n}}|^2\Big]}.
\end{align}
\setcounter{equation}{\value{mytempeqncnt}}
\hrulefill
\vspace*{4pt}
\end{figure*}
\setcounter{equation}{7}

\begin{align}\label{eq7}
\overline{\mathrm{SINR}}_{i}
=\frac{M\rho_{\mathrm{R}}}{\frac{N}{P}+(N-1)\rho_{\mathrm{R}}+1+\frac{1}{\rho_{\mathrm{R}}P}}.
\end{align}
Thus, the asymptotic SE of the $i$th RA UE without experiencing preamble collision is given by
\begin{align}\label{eq8}
\mathrm{SE}_{i}
=(1-\frac{P}{L})\log_2(1+\overline{\mathrm{SINR}}_{i}),
\end{align}
where the pre-log factor represents the preamble overhead.

With (\ref{eq1}) and (\ref{eq8}), the average SE of the $i$th RA UE is obtained as
\begin{align}\label{eq9}
\mathrm{ASE}_{i}=&{P}_{\mathrm{no\_collision}}\mathrm{SE}_{i} \nonumber\\
=&(1-\frac{1}{P})^{N-1}(1-\frac{P}{L})\log_2(1+\overline{\mathrm{SINR}}_{i}).
\end{align}

Obviously, connection between the average SE and $P$ is established in (\ref{eq9}). In the sequel, we will derive the optimal $P$ by maximizing (\ref{eq9}).

\subsection{Optimal Preamble Length}
The SE maximization problem for the $i$th RA UE is expressed as
\begin{equation}\label{eq10}
\begin{aligned}
& \underset{P}{\text{maximize}}
& & \mathrm{ASE}_{i}(P) \\
& \text{subject to} & & 1 \leq P \leq
L.
\end{aligned}
\end{equation}

To find the optimal solution of (\ref{eq10}), we first provide the
properties of $\mathrm{ASE}_{i}(P)$ in the following lemma.

\newtheorem{Lemma}{Lemma}
\begin{Lemma}
The optimal preamble length $P^{*}$ of (\ref{eq10}) must be located on the interval  $[P_1, L]$, where $P_1=\frac{-N+2+\sqrt{N^2+4N(L-1)+4-4L}}{2}$.
\end{Lemma}

\begin{proof}
With (\ref{eq9}), we rewrite $\mathrm{ASE}_{i}(P)$ as
\begin{align}\label{eq11}
\mathrm{ASE}_{i}(P)=f(P)g(P),
\end{align}
where $f(P)=(1-\frac{1}{P})^{N-1}(1-\frac{P}{L})$
and
$g(P)=\log_2\big(1+\overline{\mathrm{SINR}}_{i}(P)\big)$. Notice that, on the interval $P \in [1, L)$, 1) $f(P)>0$ and $g(P)>0$; 2)  $f'(P_1)=0$ and $g'(P)>0$;
3) $g(P)$ is strictly concave due to $g''(P)<0$.

Since $f(P)$ and $g(P)$ are both twice-differentiable, the derivative function of $\mathrm{ASE}_{i}(P)$ is given by
\begin{align}\label{eq12}
\mathrm{ASE}'_{i}(P)=f'(P)g(P)+f(P)g'(P),
\end{align}
and its second derivative is
\begin{align}\label{eq13}
\mathrm{ASE}''_{i}(P)=f''(P)g(P)+f(P)g''(P)+2f'(P)g'(P).
\end{align}

On the interval $P \in [1, P_1)$, $f'(P)>0$. With (\ref{eq12}), we have that $\mathrm{ASE}'_{i}(P)>0$.
Therefore, $\mathrm{ASE}_{i}(P)$ is strictly monotonically increasing on that interval.

On the interval $P \in [P_1, L]$, it is easy to prove that $f'(P)\le 0$ and $f''(P)< 0$. With (\ref{eq13}), $\mathrm{ASE}''_{i}(P)<0$.
As a result, $\mathrm{ASE}_{i}(P)$ is strictly concave on that interval. Moreover,
$\mathrm{ASE}'_{i}(P_1)>0$ and $\mathrm{ASE}'_{i}(L)<0$. We see that $\mathrm{ASE}_{i}(P)$ first increases and then decreases on that interval.

From these properties, the optimal preamble length $P^{*}$ must be located on the interval  $[P_1, L]$.

We conclude the proof.
\end{proof}

Based on Lemma 1, Newton's method is used to obtain the optimal preamble length $P^{*}$ of (\ref{eq10}), where $P_1$ is set to be an initial solution of (\ref{eq10}).

\newtheorem{remark}{Remark}
\begin{remark}
When $\rho_{\mathrm{R}}P^{*}\gg 1$, we consider that channels are estimated with enough accuracy. The impact of $P^{*}$ on $\mathrm{SINR}_{i}$ is thus trivial and the preamble collision and preamble overhead have the dominant impact on the SE. Therefore, $P^{*}$ is very close to $P_1$.
\end{remark}

\section{Numerical Results}
In this section, the SE of RA UE in terms of the preamble length is discussed and the
accuracy of the analyses in Section III is verified by numerical results.
Simulation parameters are summarized in Table \ref{table1}.
\begin{table}[!h]
\renewcommand{\arraystretch}{2}
\caption{Simulation parameters}\label{table1} \centering
\begin{tabular}{>{\centering}m{4cm}|>{\centering}m{4cm}|>{\centering}m{4cm}|}
\hline

\multicolumn{1}{|c|}{Number of BS antennas $M$}  & \multicolumn{2}{c|}{$100 \thicksim 500$} \tabularnewline
\hline

\multicolumn{1}{|c|}{{Number of RA UEs per Channel $N$} } & \multicolumn{2}{c|}{{$10$} } \\
\hline

\multicolumn{1}{|c|}{Payload Length $L$}  & \multicolumn{2}{c|}{$200$} \tabularnewline

\hline

\multicolumn{1}{|c|}{Uplink SNR $\rho_{\mathrm{R}}$ (dB)}  & \multicolumn{2}{c|}{$-20 \thicksim 10$} \tabularnewline

\hline
\end{tabular}
\end{table}
\begin{figure}[!ht]
\centering
\includegraphics[width=3.7in]{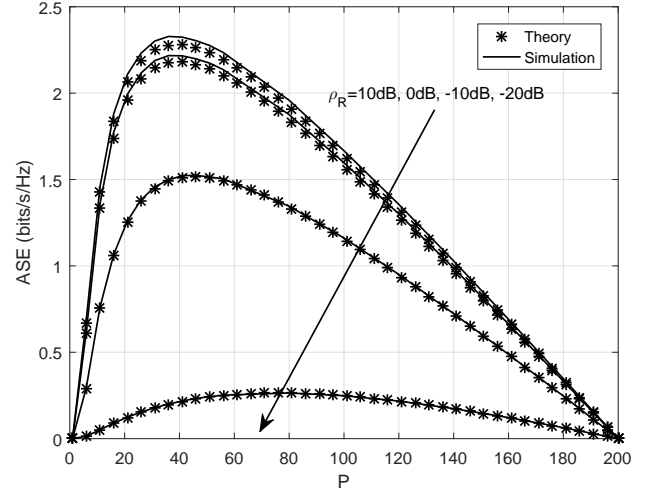}
\caption{SE of an RA UE versus the preamble length $P$, with different $\rho_\mathrm{R}$. $M=100$.} \label{fig4}
\end{figure}

In Fig. \ref{fig4}, SEs as a function of $P$ with $M=100$ are illustrated, where
the simulated SEs are compared with the corresponding analytic approximations
in (\ref{eq9}) with different  $\rho_\mathrm{R}$.
As observed,
the simulation results closely match with the analytic ones, which verify the accuracy of our derivations in Section III.
It is also noticed that each curve increases first and then decreases, and there exists an optimal value of $P$ for SE maximization, which agree with the analyses in Lemma 1.
Considering the tight agreement between the simulations and
our analyses in Section III, we adopt the analytic expression of (\ref{eq9})
in the following discussions.

\begin{figure}[!ht]
\centering
\includegraphics[width=3.7in]{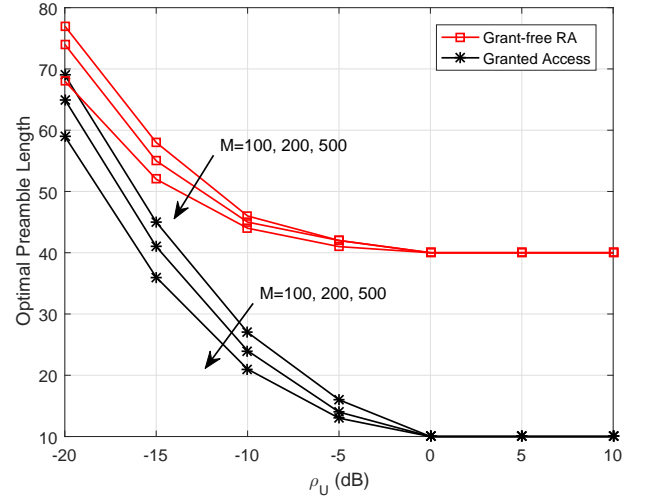}
\caption{Optimal preamble length versus $\rho_\mathrm{R}$, with different $M$.} \label{fig5}
\end{figure}

In Fig. \ref{fig5}, optimal preamble lengths for different $\rho_\mathrm{R}$ and $M$ are plotted.
In this figure, we compare the optimal preamble length in grant-free RA with that in granted access.
In the granted access, each UE has been assigned an unique preamble for uplink transmission and there is no preamble collision between UEs.
Therefore, the SE of the granted access is only affected by the channel estimation accuracy and preamble overhead. As a result, its corresponding optimal preamble length $\hat{P}^{*}$ can be achieved by maximizing (\ref{eq8}), where $ N \le P \le L$.
As shown in the figure, both optimal preamble lengths decrease to their minimum values as $\rho_\mathrm{R}$ increases. When $\rho_\mathrm{R} \ge 0$dB, $P^{*}$ keeps unchanged approximately at $P_1$ in the grant-free RA, which validates Remark 1.
Compared to the granted access, we see that longer preamble length is required for the grant-free RA due to the impact of preamble collision. Besides, it is observed that increasing $M$ has little benefit in reducing the required optimal preamble length of the grant-free RA with massive MIMO, especially when $\rho_\mathrm{R}$ is ranging from $-10$dB to $10$dB.

\begin{figure}[!ht]
\centering
\includegraphics[width=3.7in]{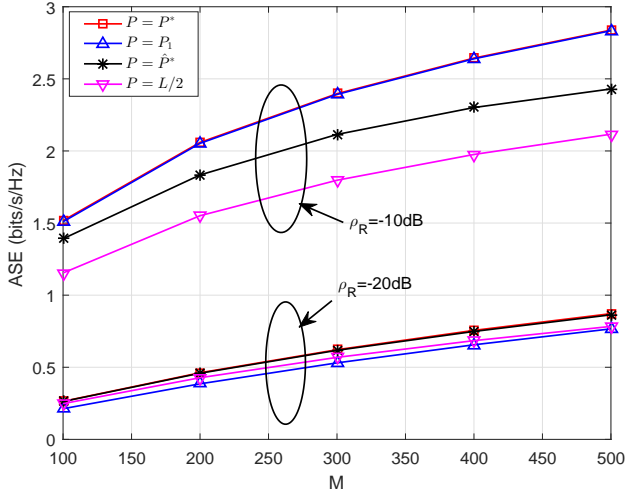}
\caption{SE of an RA UE versus the number of BS antennas $M$, with different $\rho_\mathrm{R}$.} \label{fig6}
\end{figure}

The SE of grant-free RA with respect to $M$ is shown in Fig. \ref{fig6}. Four different values of $P$ are considered: 1) the optimal preamble length in the grant-free RA with massive MIMO, i.e., $P=P^{*}$; 2) $P=P_1$; 3) the optimal preamble length in the granted access, i.e., $P=\hat{P}^{*}$; 4) $P=L/2$. Among them, it is plain to see that $P=P^{*}$ always achieves the best the SE performance. Moreover,
at $\rho_\mathrm{R}=-20$dB, $P=\hat{P}^{*}$ has the closest SE performance to that of $P=P^{*}$. This is due to the fact that, when $\rho_\mathrm{R}$ is extremely low, channel estimation accuracy and preamble overhead have more dominate impact on the SE of grant-free RA than the preamble collision and therefore
$\hat{P}^{*}$ is close to ${P}^{*}$, as shown in Fig. \ref{fig5}. On the other hand, at $\rho_\mathrm{R}=-10$dB, the SE performance of $P=P_1$ almost overlaps with that of $P=P^{*}$. This complies with the observation in Fig. \ref{fig5}, indicating that $P^{*}$ is very close to $P_1$ when $\rho_\mathrm{R}\ge -10$dB.
Combining the observations in Fig. \ref{fig5} and Fig. \ref{fig6}, it is concluded that, 1) when $\rho_\mathrm{R}$ is small, although $P^{*}>\hat{P}^{*}$, adopting $\hat{P}^{*}$ as the preamble length is enough to achieve close-to-optimal SE performance in the grant-free RA with massive MIMO; 2) when $\rho_\mathrm{R}$ is moderate or high ($\rho_\mathrm{R}\ge -10$dB), the use of $P_1$ as the preamble length is essentially optimum.

\section{Conclusions}
In this paper, we investigate the optimal preamble length for SE maximization in grant-free RA with massive MIMO. By
taking into account the effects of preamble length on the preamble collision and preamble overhead as
well as channel estimation accuracy, the analytic expression of SE in grant-free RA with massive MIMO is obtained. To maximize the SE,
optimal preamble length is analyzed and derived. Simulation results agree well with our analyses and
confirm the existence of optimal preamble length for SE maximization in the grant-free RA with massive MIMO.
Properties of the optimal preamble length with respect to system parameters are also revealed. As an example shown in the simulation,
$P^{*}  \approx P_1$ when $\rho_\mathrm{R}$ is moderate or high and $P^{*} \approx \hat{P}^{*}$ when $\rho_\mathrm{R}$ is small. 
Moreover, it is shown that longer preamble length is required for the SE maximization in the grant-free RA, compared to the granted access.



%
\normalem

\end{document}